\documentclass[aps,prl,twocolumn,superscriptaddress,groupedaddress]{revtex4-1} 
\usepackage{graphicx}  
\usepackage{dcolumn}   
\usepackage{bm}        
\usepackage{amssymb}   
\usepackage{url}
\usepackage{breakurl}

\usepackage{amsmath,amsthm,physics,complexity,fancyhdr,blindtext,array}
\usepackage[table]{xcolor} 
\usepackage{tikz}
\usetikzlibrary{quantikz}
\newtheorem{thm}{Theorem}
\newtheorem{prob}{Problem}
\renewcommand\( {\left(}
\renewcommand\) {\right)}
\def\cE{{\cal E}}
\def\cH{{\cal H}}
\def\cI{{\cal I}}
\def\cO{{\cal O}}
\def\cQ{{\cal Q}}
\def\cN{{\cal N}}
\def\cP{{\cal P}}

\begin{document}
\widetext


\title{Supersymmetry and Quantum Computation}
\author{P. Marcos Crichigno \vspace{0.2cm}}
\affiliation{Blackett Laboratory, Imperial College, 
Prince Consort Rd., London, SW7 2AZ, U.K. 
\vspace{0.1cm}}

\affiliation{
Institute for Theoretical Physics, University of Amsterdam,
Science Park 904, Postbus 94485, 1090 GL, Amsterdam, The Netherlands
} 


\begin{abstract}
The interplay between supersymmetry and classical and quantum computation is discussed. First, it is shown that the problem of computing the Witten index  of $\cN\leq 2$ quantum mechanical systems is $\#\P$-complete and therefore intractable. Then, the notions of supersymmetry in the space of qubits and supersymmetric quantum circuits are introduced and some of their properties discussed. In particular, it is shown that these define a nontrivial subclass of quantum algorithms with robustness properties typical of supersymmetric systems. Concrete examples, including the supersymmetric SYK model and fermion hard-core models are discussed. Some applications and open questions are suggested. 
\end{abstract}


\maketitle

\thispagestyle{fancy}
\rhead{Imperial/TP/2020/MC/01} 


\paragraph{Introduction.} 

Finding the ground states of a physical system is hard.  This is  true not only in a practical sense but also in a formal, computational, sense. A prototypical example is   that of  a classical spin-glass system for which the problem of finding the ground state is  $\NP$-hard \cite{Barahona_1982}. Similarly, finding the ground state of a quantum Hamiltonian is $\QMA$-hard \cite{KitaevBook,2004quant.ph..6180K}. This implies, in particular, that there is no efficient algorithm for finding ground states, assuming standard conjectures in computational complexity theory.  A number of deep connections between the theory of computational complexity and statistical mechanics systems have been pointed out in,  e.g., \cite{Kirkpatrick671,Fu_1986,NP_PhaseTransitions,aaronson2005npcomplete,CubittUndec}.

Supersymmetry is a symmetry relating bosonic and fermionic states of a system. Supersymmetric systems are often more amenable to analysis and various analytic and exact results are possible. Indeed, the ground states--or some of their properties--can be found analytically in various nontrivial supersymmetric systems. This raises the question of the computational complexity associated to the ground states of supersymmetric systems. Although many exact results are known in {\it specific} supersymmetric systems, we show  that the ground state problem for supersymmetric theories remains computationally hard.

Supersymmetry was first proposed as a possible symmetry of relativistic quantum field theory. However, its applications extend to a number of areas in mathematics including,  most famously, Morse theory, mirror symmetry, and generalized complex geometry. 

In this letter, we focus on {\it supersymmetric quantum mechanics} \cite{Nicolai:1976xp,Witten:1981nf,Witten:1982im,Witten:1982df} and bring the attention to the interplay between supersymmetry and the theory of quantum computation.  As we discuss, a natural setting for  incorporating supersymmetry in quantum computation is the fermionic model of quantum computation  \cite{2002AnPhy.298..210B,Ortiz:2000gc}. Defining  supersymmetry operators in the fermionic model, one can then map these to qubit space via a standard spin-$\tfrac12$ Jordan-Wigner transformation or its generalizations. Having defined the action of supersymmetry in the space of qubits, we then define the notion of \emph{supersymmetric quantum circuits}, and use these to design quantum algorithms associated to certain supersymmetric physical observables.

The intuition from the physics of supersymmetric systems is that this should define a nontrivial subclass of quantum algorithms, with advantageous properties (e.g., invariance under certain deformations) over non-supersymmetric ones, and which capture highly non-trivial problems of both physical as well as mathematical interest. As we discuss, this intuition bears out.

We emphasize that this definition of supersymmetry holds for any system of qubits and does not require  supersymmetry to be realized at a fundamental level in nature. In particular, any implementation of a quantum computer can be made supersymmetric in this sense. \\

\paragraph{$\cN=2$ quantum mechanics. }

The Hilbert space of any quantum mechanical theory can be decomposed as $\cH=\cH^{B}\oplus \cH^{F}$ where each factor refers to the subspace of bosonic and fermionic states. These are distinguished by the operator $(-1)^{F}$, acting as $+1$ on bosonic states and as $-1$ on fermionic states. By definition, in a theory with $\cN=2$ supersymmetry there exists a complex Grassmann  operator $\cQ$, sending states in $\cH^{B}$ into states in $\cH^{F}$ and vice versa, and satisfying the algebra \cite{Nicolai:1976xp,Witten:1981nf,Witten:1982im,Witten:1982df},
\begin{equation}\label{algebra}
H=\{\cQ,\cQ^{\dagger}\}\,,\quad  \cQ^{2}=(\cQ^{\dagger})^{2}=0\,,
\end{equation}
where $H$ is the Hamiltonian of the system and $\{(-1)^{F},\cQ\}=0$. One says an operator  is bosonic or fermionic if it commutes or anticommutes with $(-1)^{F}$, respectively. The supercharge is thus a fermionic operator and the Hamiltonian bosonic.  It follows directly from \eqref{algebra} that the spectrum of supersymmetric systems is positive semidefinite, $E\geq 0$, and that a state  $\ket{\Omega}$ has $E=0$ iff $\cQ\ket{\Omega}=\cQ^{\dagger}\ket{\Omega}=0$. All such states, which may be bosonic or fermionic, are called {\it supersymmetric ground states}. A crucial property of  states with $E>0$ is that they are paired: for every such bosonic state there is a corresponding fermionic state with the same energy. This is not necessarily the case for supersymmetric ground states. In fact, a quantity of particular interest is the Witten index, defined as the difference in the number of   bosonic and fermionic supersymmetric ground states \cite{Witten:1982df}:
\begin{equation}\label{WI}
\cI\equiv n_{E=0}^{B}-n_{E=0}^{F}=\text{Tr}_{\cH}\, (-1)^{F}\,,
\end{equation}
where in the last equality one uses the fact that  states with $E>0$ are paired and thus do not contribute to the trace. The Witten index gives a lower bound on the total number of supersymmetric ground states via the inequality $ n_{E=0}^{B}+n_{E=0}^{F}\geq \abs{\cI}$. In particular, if $\cI\neq 0$ the system must have supersymmetric ground states. 

Since $\cQ^{2}=0$, supersymmetry defines the $\Bbb Z_{2}$-graded complex of vector spaces,
\begin{equation}
C:\quad \cH^{F} \xrightarrow{\cQ}\cH^{B}\xrightarrow{\cQ}\cH^{F} \xrightarrow{\cQ}\cH^{B} \,,
\end{equation}
and the Euler characteristic of $C$ coincides with the Witten index. It is this topological nature of the Witten index that makes it a robust quantity and, in some situations, easily calculable. 

Computing general physical observables in supersymmetric systems can be as formidable a task as in non-supersymmetric systems. However, there are a subset of physical observables, ``supersymmetric observables,'' which have special properties and can often be computed exactly, the Witten index being an  example. These relate to an important set of operators called supersymmetric, or $\cQ$-closed, operators. A bosonic operator $\cO$ is $\cQ$-closed if
\begin{equation}
[\cQ,\cO]=0\,.
\end{equation}
Among these, $\cQ$-exact operators are defined  as those which can be written as $\cE=\{\cQ,\Psi'\}$, for some fermionic $\Psi'$. By nilpotency, all $\cQ$-exact operators are $\cQ$-closed.  Not all $\cQ$-closed operators, however, are necessarily $\cQ$-exact; whether or not this is the case is determined by the $\cQ$-cohomology of operators. Two supersymmetric operators $\cO$ and $\cO'$ are said to be in the same cohomology class if $\cO'=\cO+\cE$. The analogous definitions hold for fermionic $\cQ$-closed and exact operators, exchanging commutators and anticommutators. 

An important set of physical observables is given by the correlation function of supersymmetric operators in a supersymmetric ground state:
\begin{equation}\label{corr}
\langle \cO_{1}\cdots\cO_{n} \rangle_{\Omega}\equiv \bra{\Omega} \cO_{1}(t_{1})\cdots\cO_{n}(t_{n})\ket{\Omega}\,,
\end{equation}
where the $t_{i}$ are insertion points in Lorentzian time $t\in \Bbb R$, and can be expressed by a standard path integral. 

Another set of  observables is given by a refined or generalized Witten index \eqref{WI}, obtained by the insertion of supersymmetric operators into the trace:
\begin{equation}\label{indexdef}
Z_{\text{P}}[\cO_{1}\,,\cdots \,,\cO_{n}]\equiv \text{Tr}_{\cH}\, \left[(-1)^{F}\cO_{1}(\tau_{1})\cdots \cO_{n}(\tau_{n})\right]\,.
\end{equation}
This can be thought of as the insertion of operators in the Euclidean path integral of the theory, with periodic boundary conditions for fermions along a compactified Euclidean time direction $\tau=i t$. An important  property of the observables \eqref{corr} and \eqref{indexdef} is that they are invariant under exact deformations,
\begin{equation}\label{exactdef}
\cO_{k}\to \cO_{k}+\cE_{k}\,,
\end{equation}
as can be easily checked. For the former, this follows from properties of the supersymmetric ground state and, for the latter, from cyclicity of the trace.  Thus, these observables are sensitive only to the cohomology class of supersymmetric operators. These robustness properties will be relevant to our discussion of quantum computation below.  \\

\paragraph{Computational complexity of supersymmetric systems.} 
Let us briefly review relevant concepts of complexity theory  (see, e.g., \cite{sipser13,PapdimitriouBook}). The complexity class $\P$ is the class of decision problems (with a ``yes/no'' answer) which can be solved by a deterministic Turing machine in polynomial time. The class $\NP$ is the class of decision problems for which the problem instances which give ``yes''  can be \emph{checked} in polynomial time.  The complexity class $\#\P$ is the set of \emph{counting} problems associated to decision problems in $\NP$. For example, whereas the problem of deciding if a boolean formula has a satisfying instance is a problem in $\NP$, the problem of counting how many satisfying instances it has is a problem in $\#\P$. A problem {\sc H}  is said to be $\#\P$-hard  if it is at least as hard as any problem in $\#\P$ or, more precisely, if any  problem in $\#\P$ can be reduced to {\sc H} in polynomial time. A problem is said to be $\#\P$-complete if it is  $\#\P$-hard  and belongs to the class $\#\P$. 

The  Witten index \eqref{WI} can sometimes be computed exactly and with little dynamical information. In particular, supersymmetry ensures that the contribution of all  states with $E>0$ in the trace in \eqref{WI} cancel out and thus the index  can  be computed with no knowledge of  the supersymmetric ground states themselves. Furthermore, under certain conditions the index is invariant under small, supersymmetric, deformations of the system \cite{Witten:1982df} which can sometimes be  exploited to  bring the system to a weakly coupled point, where the Witten index can be efficiently computed in perturbation theory. Although this is often the case in \emph{specific} supersymmetric systems, we show next there can be no efficient algorithm for computing the Witten index for generic supersymmetric systems, assuming standard conjectures in computational complexity. 

For the purposes of studying the complexity of the problem we consider the Hilbert space $\cH$ to be spanned by a {\it subset} of $N$-bit strings $\ket{n_{1},\cdots,n_{N}}$, that this subset is determined by a polynomial number of constrains among the $n_{i}$, and that the supercharge is a $k$-local function on $\cH$. Then, one can prove the following Theorem:

\begin{thm}\label{thmWI}
Given a quantum mechanical system with a finite-dimensional Hilbert space $\cH\subseteq (\Bbb C^{2})^{\otimes N}$ and $\cN\leq 2$ supersymmetry,  specified by a polynomial number of constraints and a $6$-local supercharge $\cQ$, the problem of computing the Witten index is $\#\P$-complete.  
\end{thm}
(See the Supplemental Material for a proof.) This implies, in particular, that finding ground states of supersymmetric systems is intractable.

Having discussed a consequence of the theory of computation in supersymmetric systems we now discuss some consequences of supersymmetry  in the theory of quantum computation. \\

\paragraph{Supersymmetry in qubit space.} The Hilbert space of $N$ qubits comes with a natural $\Bbb Z_{2}$-grading, given by strings with an even and odd number of 1's. The basic observation we make here is that these can be consistently identified with ``bosonic'' and ``fermionic'' subspaces, respectively, with $+1$ and $-1$ parity under $(-1)^{F}$ with $F$ the Hamming weight.  In general, we consider a Hilbert space $\cH\subseteq(\Bbb C^{2})^{\otimes N}$ and define $\cN=2$ supersymmetry in the space of qubits as  a nilpotent map $\cQ$ acting on $\cH$ and sending states with even parity into odd parity and vice versa.  

One way to construct  a  map $\cQ$ is to recall that the Hilbert space of $N$ qubits is isomorphic to the Hilbert space of $N$ (spinless) fermions, which is exploited in the fermionic model of quantum computation \cite{2002AnPhy.298..210B,Ortiz:2000gc}. In this model one considers $N$ vertices of a graph $G$, each of which can be occupied by 0 or 1 spinless fermions. A fermion at vertex $i$ is created by an operator $a_{i}^{\dagger}$ and annihilated by $a_{i}$, satisfying
\begin{equation}
\{a_i,a_j^\dagger\}= \delta_{ij}\,,\qquad i,j\in\{1,\ldots,N\}\,,
\end{equation}
and all other anticommutators vanishing. The $2^{N}$-dimensional Fock space is constructed by acting  with creation operators on the vacuum state with no fermions, $\ket{0\cdots0}_{f}$,  and is in one-to-one correspondence with the space of $N$ qubits:
\begin{equation}
\ket{n_{1}\cdots n_{N}}_{q}\leftrightarrow (a_{1}^{\dagger})^{n_{1}}\cdots (a_{N}^{\dagger})^{n_{N}}\ket{0\cdots0}_{f}\,.
\end{equation}
Operators in qubit space are obtained from operators in Fock space via a spin-$\tfrac12$ Jordan-Wigner transformation,
\begin{equation}\label{asigma}
a_i \rightarrow K_{i}\,  \sigma_+^i \,, \quad a_i^\dagger  \rightarrow \sigma_-^i K_{i}^{\dagger}  \,,
\end{equation}
where $\sigma_{\pm}=\tfrac12(\sigma_{x}\pm i \sigma_{y})$ and the $K_{i}$ are non-local operators, which depend on the graph. In the case of a 1d graph $K_{i}=\prod_{j=1}^{i-1}(-1)^{n_{i}}$ (see, e.g., \cite{Ortiz:2000gc} and references therein for generalizations). The first example of  a system of spinless fermions on a graph with $\cN=2$ supersymmetry was constructed by Nicolai \cite{Nicolai:1976xp}, in which the supercharges are cubic functions of the creation and annihilation operators. As a generalization, we consider the  ansatz,
\begin{equation}\label{ansatzQ}
\cQ =  \sum_{i} a_{i}^{\dagger} \, B_{i}(a,a^{\dagger})\,, \quad \cQ^{\dagger} =  \sum_{i}  B^{\dagger }_{i}(a,a^{\dagger})\, a_{i}\,,
\end{equation}
where the $B_{i}$ are a set of bosonic operators built out of the creation/annihilation operators.  Nilpotency  requires
\begin{equation}\label{nilpQ}
\cQ^{2}= \sum_{i,j} a_{i}^{\dagger} B_{i}(a,a^{\dagger})\, a_{j}^{\dagger} B_{j}(a,a^{\dagger})=0\,.
\end{equation}
Different solutions to this constraint amount to different realizations of supersymmetry on a system of $N$ fermions. There are various interesting explicit models including the fermion hard-core model of \cite{Fendley:2002sg}, relevant to the study of the clique or independence complex, and  the supersymmetric SYK model \cite{Fu:2016vas}, relevant to holography.  (See the Supplemental Material for more details and another method for constructing and classifying explicit supercharge representations.) For now, we keep the discussion general and do not specify the choice of supercharge or $\cH$.  Now, applying  \eqref{asigma}   we have 
\begin{equation}\label{QQbits}
\cQ =\sum_i  \sigma_-^i K_{i}^{\dagger}\,B_{i}\,, \quad \cQ^\dagger =\sum_i  B_i^\dagger\, K_{i}\sigma_+^i \,,
\end{equation}
where $B_{i}=B_{i}(K\sigma_{+},\sigma_{-}K)$. This defines the action of $\cN=2$ supersymmetry on the space of qubits.

Note that the supersymmetry operators \eqref{QQbits} could have been defined directly in the space of qubits, with no reference to the fermionic model. Thus, although the fermionic system does not play a fundamental role, it  provides a  setting where  supersymmetry is naturally defined. \\

\paragraph{Supersymmetric circuits.}

 Let us call a circuit $U$ bosonic if it preserves the parity of the state it acts on,  $[(-1)^{F},U]=0$, and  fermionic if it flips it, $\{(-1)^{F},U\}=0$. The notions of closedness and exactness in  Fock space translate directly into the corresponding notions in qubit space.  We  define a bosonic \emph{supersymmetric quantum circuit} as a bosonic quantum circuit $U_{S}$, which is closed with respect to the  supercharge $\cQ$, i.e.,  
\begin{equation}\label{QU}
[\cQ,U_{S}]=0\,.
\end{equation}
A fermionic supersymmetric circuit is similarly a fermionic circuit satisfying $\{\cQ,U_{S}\}=0$. An obvious example of a bosonic supersymmetric circuit  is time evolution by the supersymmetric Hamiltonian, as $U_{S}=e^{- i t \{\cQ,\cQ^{\dagger}\}}$ commutes with  $\cQ$ (and also with $\cQ^{\dagger}$ in this special case) and with $(-1)^{F}$.  Note the composition of  supersymmetric circuits by matrix multiplication is supersymmetric. In fact, it is straightforward to see that, for a given $\cQ$, the collection of supersymmetric circuits form a group.  Similarly, we call a qubit state $\cQ$-closed if $\cQ\ket{s}_{q}=0$ and a supersymmetric ground state if it is  closed with respect to both supercharges, $\cQ\ket{\Omega}_{q}=\cQ^{\dagger}\ket{\Omega}_{q}=0$. 
From now on we drop the subscript $q$, with the understanding that all states refer to qubit states. 

Thus, for a given supercharge $\cQ$, Eq.~\eqref{QU} imposes a constraint on the class of quantum circuits we consider. The next question is which states we allow as inputs into supersymmetric circuits. Here we let the supersymmetric observables reviewed above guide us, which suggest two natural ``modes'' of computation. \\

\paragraph{The supersymmetric Hadamard test.}
\begin{figure}
\centering
\includegraphics{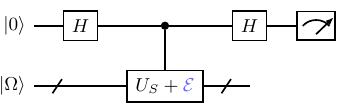}
\caption{The Hadamard test for a supersymmetric circuit. If the input states are supersymmetric ground states $\ket{\Omega}$, the deformation $\cE$ does not affect the outcomes of the ancilla qubit provided $\cE= \{\cQ,\Psi\}$.  }
\label{fig:Had_Circuit}
\end{figure}
The first mode of computation, suggested by the  correlation function \eqref{corr}, is to take as input supersymmetric ground states and calculate matrix elements of  $U_{S}$. The standard technique for estimating  matrix elements is the Hadamard test \cite{aharonov2005polynomial} which, given a state $\ket{s}$ and a unitary matrix $U$, estimates $\bra{s}U\ket{s}$. The procedure consists of adding  an ancilla qubit  initialized to $\ket{0}$, then Hadamarded to $\ket{+}$, and  used as a control qubit for unitary evolution of the state $\ket{s}$ by  the circuit $U$. Finally, the ancilla qubit is Hadamarded one more time and measured in the computational basis. Applying this to  a supersymmetric circuit  in a supersymmetric ground state, the probability of measuring $\ket{0}$ in the ancilla qubit is
\begin{equation}
\label{p0Had}
p(0) = \frac{1}{2}\(1+\Re \bra{\Omega}U_{S}\ket{\Omega}\)\,.
\end{equation}
Repeated measurements of the ancilla qubit then leads to an estimate of $\Re \bra{\Omega}U_{S}\ket{\Omega}$. (The imaginary part is  obtained by changing the initial state of the ancilla.) What is particular about the supersymmetric setting is that the robustness property of supersymmetric systems translates directly into a corresponding property of the quantum circuit. Namely, the  outcome probability \eqref{p0Had}  is unaffected by a deformation of the circuit of the form
 \begin{equation}\label{exactdefqubit}
 U_{S}\to U_{S}+\cE\,,
\end{equation}
with $\cE=\{\cQ,\Psi\}$ for some $\Psi$ (see Fig.~\ref{fig:Had_Circuit}). More generally, for $U_{S}=U_{S}^{(1)}\cdots U_{S}^{(n)}$ a product of supersymmetric circuits the outcome probability is unaffected by an exact deformation of any of the circuits, $U_{S}^{(k)}\to U_{S}^{(k)}+\cE^{(k)}$.  This is the first example of ``robustness'' of supersymmetric quantum circuits. If one is interested in computing correlation functions of operators $\cO$, not necessarily unitary as in \eqref{corr}, one can expand each operator into a linear combination of unitaries,  as in \cite{Ortiz:2000gc},  and then apply the algorithm above to each term.

To run this algorithm the supersymmetric ground states  $\ket{\Omega}$ must be either known in advance or  be prepared. However, since the question of whether  a system has supersymmetric ground states to begin with is hard, we would like to design an algorithm that  gives us  information on possible supersymmetric ground states.  \\

\paragraph{The generalized Witten index algorithm.} The second mode of computation, suggested by the refined Witten index  \eqref{indexdef}, is to allow for mixed states and take this to be the maximally mixed state, $\rho_{\cH}=\mathbb{I_{\cH}}/\dim \cH$. This  has the advantage that there is no need to prepare supersymmetric ground states.
\begin{figure}
\centering
\includegraphics{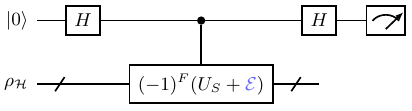}
\caption{The trace estimation algorithm for the generalized Witten index $\text{Tr}\left[(-1)^{F}U_{S}\right]$. The average outcome of the ancilla qubit is  insensitive to exact deformations $\cE$.    }
\label{fig:WittenAlgoDef}
\end{figure}

The standard method for calculating the trace of a unitary matrix acting on $M$ qubits is trace estimation \cite{PhysRevLett.81.5672}. This is basically the same as the Hadamard test above except that, rather than inputing a pure state $\ket{s}$, the input is $\rho_{\cH}$. This is known as the one-clean-qubit model of quantum computation \cite{PhysRevLett.81.5672}. One then applies a controlled unitary evolution of the density matrix, $\rho_{\cH}\to U \rho_{\cH}U^{\dagger}$, and measures the ``clean'' ancilla qubit.  Applying this to a circuit of the form $U=(-1)^{F}U_{S}$, the probability of measuring $\ket{0}$ in the ancilla qubit is given by
\begin{equation}\label{p0}
p(0)= \frac12\(1+\frac{1}{2^{M}} \Re \text{Tr}_{\cH} \left[ (-1)^{F} U_{S}\right]\)\,,
\end{equation}
and repeated measurement of the ancilla qubit approximates the generalized Witten index. We refer to this as the {\it generalized quantum Witten machine}. Once again, we note that the outcome probability  \eqref{p0} is unaffected by the exact deformation \eqref{exactdefqubit} (see Fig.~\ref{fig:WittenAlgoDef}).  More generally, for $U_{S}$ a product of supersymmetric circuits the outcomes are unaffected by an exact deformation of each circuit. If one is interested in $\text{Tr}_{\cH}\,[(-1)^{F}\cO]$ with $\cO$ not necessarily unitary one may expand $\cO$ as sum of unitary supersymmetric matrices. \\

Another interesting property arises   when the circuit is closed with respect to both supercharges. Consider the supersymmetric circuit $U_{S}=\hat U_{S} U_{\cE}$, with $[\cQ,\hat U_{S}]=[\cQ^{\dagger},\hat U_{S}]=0$ and $U_{\cE}\equiv e^{i \text{Re}\, \cE}$. Then, one can check that 
 \begin{equation}
\text{Tr}_{\cH} \left[(-1)^{F} \hat U_{S}\, U_{\cE}\right] = \text{Tr}_{\cH} \left[ (-1)^{F} \hat U_{S}\right]\,.
\end{equation}
Thus, circuits of the form of $U_{\cE}$ can be completely removed from such quantum algorithms (see  Fig.~\ref{fig:WittenAlgoExactCirc}). 
\begin{figure} 
\centering
\includegraphics{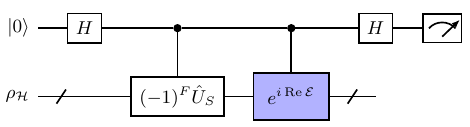}
\caption{ The unitary circuit  $e^{i\Re \cE}$ can be dropped from the algorithm, without affecting the measured outcomes, provided the input state $\rho_{\cH}$ is the maximally mixed state  and $[\cQ,\hat U_{S}]=[\cQ^{\dagger},\hat U_{S}]=0$. }
	\label{fig:WittenAlgoExactCirc}
\end{figure}

It is important to note that the algorithm just described provides only an  approximation to the (generalized) Witten index and \emph{not} an exact result which, as  implied by Theorem~\ref{thmWI}, is intractable even for a quantum computer. Indeed, for the trace estimation algorithm to provide an $\cO(1)$ approximation to the trace, an exponential number of measurements is required. \\

\paragraph{Outlook.}  

An interesting question is whether the local Hamiltonian problem remains $\QMA$-hard for  supersymmetric systems.  A physical intuition on why this  may be the case  is that  supersymmetric models can exhibit frustration \cite{Huijse_2008}. Since supersymmetric ground states are elements of cohomology we expect this to be relevant to the complexity of computing cohomology groups.  Similarly, one may ask if trace estimation remains $\mathsf{DQC1}$-hard \cite{shor2007estimating} for matrices of the form $U=(-1)^{F}U_{S}$, with $U_{S}$ supersymmetric. If so, this would strongly suggest that the generalized Witten machine cannot be simulated classically and thus gives a quantum advantage. We hope to return to these questions in the future.

We note that the robustness properties discussed above suggest the possibility that these may be exploited to design algorithms that are intrinsically robust under certain ``errors.''  The crucial point here is that the supercharge $\cQ$ is not dictated by the environment but rather one is free to choose it. Then, given an error $\cE$ one wishes to protect against, the question is  whether there is a choice of supercharge $\cQ$ and an operator $\Psi$ such that $\cE=\{\cQ,\Psi\}$. If so, and the $\cQ$-cohomology of unitary operators is nontrivial, this would lead to a family of quantum algorithms which are robust under these errors.   

Interesting computational problems also arise in continuum supersymmetric quantum mechanical systems, including the calculation of the Morse index of a function  and the Euler characteristic of a manifold \cite{Witten:1982im,Witten:1982df} (see also \cite{Hori:2003ic} for  various applications of supersymmetric quantum mechanics). It would be interesting if some of the ideas presented here could be applied to a discretized version of these systems. \\

Finally, we hope that our discussion brings the attention to the  role of supersymmetry in quantum computation and quantum information more broadly, a subject much underexplored. 

\

\begin{acknowledgments} The author thanks Jan de Boer, Chris Cade, Irina Kostitsyna,  Kareljan Schoutens, and Ronald de Wolf for discussions. This work was  supported by Nederlandse Organisatie voor Wetenschappelijk Onderzoek (NWO) via a Vidi grant and is also part of the Delta ITP consortium, a program of the NWO that is funded by the Dutch Ministry of Education, Culture and Science (OCW), and  by the EU's Horizon 2020 Research Council grant 724659 Massive-Cosmo ERC-2016-COG and the STFC grant ST/T000791/1.
\end{acknowledgments} 


\bibliographystyle{apsrev4-1} 

\newpage 
\null 

\newpage 

\section{\large Supplemental Material}
\label{sec:suppmap}

\subsection{Constructing supercharges in qubit space}

Here we describe an algebraic method for classifying and constructing explicit representations of $\cN=2$ supersymmetry in the space of $N$ qubits. It is known that conjugacy classes of   nilpotent  degree $2$ matrices of size $M\times M$   are in one-to-one correspondence with partitions of $M$ of size up to $2$, i.e., 
\begin{equation}
p(M,2) = (\underbrace{2,\cdots,2}_{n_{2}},\underbrace{1,\cdots,1}_{n_{1}})
\end{equation}
with $2n_{2}+n_{1}=M$, so $n_{2}=\{0,1,\cdots,\frac{M}{2}\}$ and there are $\frac{M}{2}+1$ such partitions. These can be thought of as Young tableau  with a total of $M=2^{N}$ boxes, where each row  consists of either 2 or 1 boxes. For each partition one can define a nilpotent operator, which takes a Jordan normal form with the partition indicating the number of $2\times 2$  and $1\times 1$ Jordan blocks. The partition  $(1,\cdots,1)$ is trivial since the matrix vanishes. Thus, there are a total of $M/2$ nontrivial conjugacy classes. For example, for a two qubit system $M=2^{2}=4$, the nontrivial partitions are $(2,2)$ and $(2,1,1)$ and the two nontrivial nilpotent matrices in Jordan normal form are, respectively, 
\begin{equation}
\tilde \cQ_{1}=
\begin{pmatrix}
 0 &  1 & 0 & 0 \\
 0 & 0 & 0 & 0 \\
0 & 0 & 0 & 1 \\
0 & 0 & 0 & 0
\end{pmatrix}\,,\qquad \tilde  \cQ_{2}=\begin{pmatrix} 0 & 1 &0 &0 \\ 0& 0&0&0 \\ 0& 0&0&0 \\0& 0&0&0\end{pmatrix}\,.
\end{equation}
This gives two conjugacy classes, given by  $\cQ_{i}=S^{-1}\tilde \cQ_{i}S$ with $S$ an invertible matrix. One can see that the $(2,2)$ partition corresponds to the conjugancy class of $\cQ=a_{1}^{\dagger}+a_{2}^{\dagger}$ and the partition $(2,1,1)$ to the conjugancy class of the fermion hard-core model $\cQ=a_{1}^{\dagger}P_{1}+a_{2}^{\dagger}P_{2}$. As discussed in \cite{Witten:1982df}, if $S$ is a unitary matrix the supercharges $\cQ_{i}$ and $\tilde \cQ_{i}$ are physically equivalent, differing only by a change of basis in Hilbert space. For $S$ invertible but non-unitary this conjugation leads to a family of physically inequivalent supercharges. Nonetheless,  there is a one-to-one correspondence between supersymmetric ground states of $\cQ_{i}$ and $\tilde \cQ_{i}$ and thus, in particular, the Witten index is unchanged.  Finally, we note that in this construction of supercharges nilpotency  is ensured but anticommutation with  $(-1)^{F}$ has to be imposed at the end (after conjugating by $S$). In contrast,  in the approach in terms of creation and annihilation operators,  anticommutation with $(-1)^{F}$  is ensured and the nilpotency condition \eqref{nilpQ} is a nontrivial constraint.

\subsection{Proof of Theorem~\ref{thmWI}}

\begin{proof}
To show that the Witten index problem for $\cN=2$ quantum mechanics is $\#\P$-hard, it is enough to show that it is hard in a specific instance. To this end we consider the  fermionic hard-core model introduced in \cite{Fendley:2002sg}. This model consists of fermions on a graph $G$ with a strong repulsive core, forbidding  a fermion to occupy a site $i$ if any of its adjacent sites $j$ is occupied. Thus, a basis for the Hilbert space is given by all subsets of vertices of $G$ such that no two vertices are adjacent, known as  independent sets of $G$. This can be summarized by giving $\cO(N^{2})$ number  constraints on the basis elements, i.e., 
\begin{equation} \label{Hindep}
\cH=\text{span}\{\ket{n_{1},\cdots, n_{N}}\,|\,A_{ij}n_{i}n_{j}=0\}\,,
\end{equation}
where $A_{ij}$ is the incidence matrix of $G$. The supercharge is given by \cite{Fendley:2002sg}
\begin{equation} \label{Qindep}
\cQ=\sum_{i}a_{i}^{\dagger}P_{i}\,,\qquad  P_{i}\equiv\prod_{j\to i }(1-\hat n_{j})
\end{equation}
where the product is over all vertices $j$ adjacent to vertex $i$ and $\hat n_{i}= a_{i}^{\dagger}a_{i}$. The projectors $P_{i}$ ensure that the supercharge preserves the Hilbert space \eqref{Hindep}.

\begin{figure}
\includegraphics{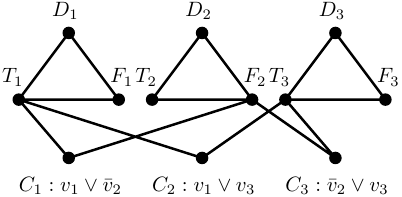}
\caption{Gadget used in \cite{ROUNE2013170} to reduce $\#\SAT$ to the problem of computing the Euler characteristic of the independence complex. To each variable $v_{i}$ there is an associated triangle and for each clause $C_{a}$ there is an associated vertex and these are connected according to whether a variable enters negated or not in the clause (see \cite{ROUNE2013170} for details). In our setting, this corresponds to computing the Witten index of an $\cN=2$ quantum mechanical system with a $6$-local supercharge $\cQ$. }
\label{graphSAT}
\end{figure}

This system defines a particular chain complex known as the independence complex (see e.g., \cite{huijse2009supersymmetry,Huijse:2011aa} and references therein). Thus the problem of computing the Witten index of this system amounts to computing the Euler characteristic of the independence complex. The latter  was shown to be $\#\P$-hard in \cite{ROUNE2013170}, given $G$ as an input, by reduction from $\#$2-$\SAT$ (see Fig.~\ref{graphSAT}). In our setting, instead, the input is the Hilbert space \eqref{Hindep} and supercharge  \eqref{Qindep}. However, since these are  efficiently computable from $G$ it follows that computing the Witten index is $\#\P$-hard. Note that the supercharge \eqref{Qindep} is $(\delta+1)$-local, with $\delta$ the maximum degree of $G$. Although the construction in  \cite{ROUNE2013170} in general leads to a graph with a degree $\delta=\text{poly}(N)$, it is easy to see that with a polynomial amount of work it can be brought into a graph of maximum degree $\delta=5$, resulting in a $6$-local supercharge, while preserving the corresponding Witten index. To see this, note that the degree of the graph in \cite{ROUNE2013170} is given by $\delta=s+2$, with $s$ the maximum number of occurrences of any variable in the clauses of the original $2$-$\SAT$ problem (see Fig.~\ref{graphSAT}). Now, given any instance of $\#2$-$\SAT$ one can reduce the maximum number of appearances of each variable down to $s=3$ by the construction of \cite{TOVEY198485}. Namely, for each variable $x$ which appears in $l> 3$ clauses one introduces $l$ new variables $x_{1},\ldots,x_{l}$, replaces the $l$th occurrence of $x$ with $x_{i}$, $i=1,\ldots, l$, and appends the clauses $(x_{i}\lor \bar x_{i+1})$ for $i=1,\ldots,l-1$ and $(x_{l}\lor  \bar{ x}_{1})$. These force the variables (for each $x$) to be the same, $x_{i}=x$, thus recovering the original instance of $\#\SAT$. However, each variable appears no more than 3 times in the extended set of clauses, which only introduces a polynomial number of new variables. Thus, without loss of generality we can assume $s=3$ and the corresponding maximum degree of the graph is $\delta=5$. The supercharge is thus $6$-local, as claimed. 

What remains to be shown is that the  Witten index problem is in $\#\P$. One can easily show this by adapting an argument in \cite{ROUNE2013170} to our setting. First, one uses the property \eqref{WI} to write the Witten index as the difference in the {\it total} number of bosonic and fermionic states in the Hilbert space, i.e., $\cI=n^{B}-n^{F}$. Although the problem of computing each term, $n^{B}$ and $n^{F}$, separately is in $\#\P$, it does not follow from this that computing the difference is necessarily in $\#\P$. However, using the identity
\begin{equation}
n^{F}+\bar n^{F} = 2^{N-1}\,,
\end{equation}
where $\bar n^{F}$ denotes the number of fermionic states  in $\(\Bbb C^{2}\)^{\otimes {N}}$ but {\it not} in $\cH$, one has
\begin{equation}
\cI+2^{N-1}=n^{B}+\bar n^{F}\,.
\end{equation}
Then, one  considers the decision problem: Given a  Hilbert space $\cH$ defined by a polynomial number of constraints, is there a bosonic state satisfying the constraints or a fermionic state not satisfying the constraints? If so, this is a ``yes'' instance of the problem. This is a problem in $\NP$, with a witness for a ``yes'' instance the corresponding bosonic or fermionic state. Thus, the corresponding counting version of this problem is in $\#\P$ and, since one can subtract $2^{N-1}$ efficiently, the Witten index problem is in $\#\P$. Thus, the Witten index problem is $\#\P$-complete. 

\end{proof}

 As a corollary, no efficient algorithm for finding or counting all bosonic and fermionic ground states of $\cN\leq 2$ systems is expected to exist either; if it did exist the Witten index could then be computed with this information in polynomial time. It also  follows  that  the problem of computing the more general \eqref{indexdef} is $\#\P$-hard. 

\subsection{Generalized Witten index in the hard-core model}

 We now define the problem {\sc Generalized Witten Index-additive} and discuss an algorithm in the case of the hard-core model.  Consider a supersymmetric operator of the form $\cO=e^{\mu J}$ where $J=J^{\dagger}$ is a Hermitian supersymmetric operator, $[\cQ,J]=0$, and $\mu\geq 0$. Note that the supersymmetry algebra implies $[J,H]=0$ and thus $J$ is a symmetry of the system  with $\mu$ the associated chemical potential.  One can define the  generalized Witten index,
\begin{equation}\label{genWA}
Z_{\text{P}}[\mu]\equiv\, \text{Tr}_{\cH}\left[ (-1)^{F}e^{\mu J} \right]\,.
\end{equation}
This function receives contributions only from $\cQ$-closed states and is thus a generating function, counting (with signs) the number of $\cQ$-closed states with a given quantum number with respect to $J$. As already noted, we do not expect the existence of an efficient algorithm, classical or quantum, computing this.  Indeed, the special case $J=H$ corresponds to the Witten index, which is $\#\P$-complete. We therefore focus on developing an approximation and define the following problem:
 \begin{prob}[\sc Generalized Witten Index-additive]
Given a supercharge $\cQ$ acting on $N$ qubits, a Hermitian supersymmetric operator $J$, bounded below by $\lambda$, and two numbers $\epsilon>0$ and $1/2<c<1$, output a number $\hat Z_{\text{P}}$ such that
\begin{equation}\label{probWitten}
\abs{\hat Z_{\text{P}}-\frac{1}{2^{N}e^{\lambda \mu}}Z_{\text{P}}[\mu]} \leq \epsilon \,,
\end{equation}
with probability greater than $c$.
\end{prob}
This is very closely related to the partition function problem defined in \cite{fern2008entanglement}, with some subtle differences. The first is the presence of $(-1)^{F}$, whose role is to impose periodic boundary conditions for fermions (as opposed to the more standard antiperiodic boundary conditions), and the second is that the trace is not carried over the full space of computational states but over the subspace $\cH \subseteq \(\Bbb C^{2}\)^{\otimes {N}}$ which can, in general, be hard to  find. We now show how these differences can be  overcome in the case of the fermion hard-core model, showing that the algorithms of  \cite{chowdhury2019computing} can be easily adapted to providing an additive approximation to \eqref{genWA}. 

Given the fermion hard-core model discussed above,  we wish to compute
\begin{equation}
Z_{\text{P}}[\mu]=\sum_{s\in \{0,1\}^{N}} (-1)^{F} e^{\mu J}\cP\,,
\end{equation}
where $\cP$ is a projector onto independent sets of $G$. Since we are interested in an approximation, the strategy is to relax the strict projector $\cP$ and define the quantity
\begin{equation}\label{Wittengamma}
\xi(\mu;\gamma)\equiv  \sum_{s\in \{0,1\}^{N}} (-1)^F e^{\mu J} e^{-\gamma H_{\mathit Pen}}\,,
\end{equation}
where $ \gamma>0$ is a parameter and  $H_{\mathit Pen}\geq 0$ is a  ``penalty Hamiltonian,'' which is designed so that
\begin{equation}\label{condHpen}
H_{\mathit Pen}\ket{s}= \varepsilon \ket{s}\ \quad \text{with} \quad \begin{cases} \varepsilon=0 \qquad \text{if $\ket{s}\in \cH(G)$} \\ \varepsilon>0 \qquad \text{otherwise} \end{cases}\,.
\end{equation}
This assigns weight 1 to  states respecting the hard-core condition  and exponentially suppresses states violating it.  
Note that we are now allowing for arbitrary states in $\(\Bbb C^{2}\)^{\otimes {N}}$ and we have thus enlarged the space in which computation can be performed, from the reduced confine of  independent sets to the full $2^{N}$-dimensional computational space. States which are not independent sets, however, are penalized and their contribution is exponentially suppressed. The quantity \eqref{Wittengamma} may be seen as a grand canonical partition function,  $\mathcal Z=\sum_{s\in \{0,1\}^{N}} e^{\nu F} e^{-\beta H} $, with $\beta H = \gamma H_{pen}-\mu J$, and complexified value of the chemical potential $\nu=i \pi$. A simple choice for the penalty Hamiltonian is
\begin{equation}\label{Hpensimple}
H_{\mathit Pen}=  \sum_{i=1}^N \sum_{j\to i} \hat n_i \hat n_{j} \,,
\end{equation}
which satisfies \eqref{condHpen} and has the property that it is $\cQ$-closed but not $\cQ$-exact. Using the fact that $J$ has a lower bound $\lambda$, i.e., that $\lambda\leq  \lambda_{i}$ for all eigenvalues $\lambda_{i}$ of $J$ it follows that, to leading order in $N$, 
\begin{equation}
-2^{N}e^{\mu \lambda}e^{-\gamma}\leq \xi(\mu;\gamma)-Z_{\text{P}}[\mu]\leq2^{N}e^{\mu \lambda}e^{-\gamma}\,,
\end{equation}
and choosing $\gamma=\log (2/\epsilon)$, we obtain the inequality
\begin{equation}\label{approxW}
\frac{1}{2^{N}e^{\mu\lambda}}\abs{ \xi(\mu;\gamma)-Z_{\text{P}}[\mu]} \leq \frac{\epsilon}{2}\,.
\end{equation}
Thus, the (complexified) partition function \eqref{Wittengamma} gives an additive approximation to the generalized Witten index, within a window $\epsilon/2$. Now, consider another quantity, $\hat  \xi$, which itself approximates a normalized \eqref{Wittengamma},  to an additive accuracy $\epsilon/2$, namely
\begin{equation}
\abs{\hat  \xi(\mu;\gamma)-\frac{1}{2^{N}e^{\mu\lambda}} \xi(\mu;\gamma)}\leq \frac{\epsilon}{2}\,.
\end{equation}
Combining this with \eqref{approxW} implies
\begin{equation}\abs{\hat  \xi(\mu;\gamma)-\frac{1}{2^{N}e^{\mu\lambda}}Z_{\text{P}}[\mu]}\leq \epsilon\,,
\end{equation}
as required in \eqref{probWitten}. Thus, the problem of obtaining an additive approximation to the generalized Witten index of the hard-core model  is reduced to that of finding an additive approximation to the (complexified) partition function \eqref{Wittengamma}. 
To find such an approximation  one can apply the algorithm of \cite{chowdhury2019computing}. Following this approach, one expands  $ e^{-\beta H}=e^{-(\gamma H_{\mathit Pen}-\mu J)}$  as a sum of unitary operators, to each of which trace estimation is applied, with the only difference that each term in the expansion is multiplied by the additional  unitary operator $(-1)^{F}$. Since this is a very simple operator it leads only to a polynomial overhead. Note also  that $H_{\mathit Pen}$ is a sum of 2-local Hamiltonians which, based on the results in \cite{chowdhury2019computing}, can be dealt with efficiently. The complexity of the algorithm is then dominated by that of approximating  $e^{\mu J}$ which under certain conditions, as the ones detailed in \cite{chowdhury2019computing,fern2008entanglement}, can be made efficient. In the case $\mu=0$ this gives an efficient approximation to the Witten index. However, this provides no advantage over classical algorithms as a classical sampling algorithm can give an estimate to the same accuracy \footnote{The author thanks Chris Cade for discussions on this. }.  More generally, this model  provides a  setting in which the question of  $\mathsf{DQC1}$-hardness for $\mu \neq 0$ may be addressed. Generalizations of this model, in which up to $k$ consecutive vertices may be occupied,  denoted  $M_{k}$,  were introduced in \cite{Fendley_2003}. Interestingly, these describe a  discretization of  $\cN=2$ superconformal minimal models at level $k$ \cite{SciPostPhys.3.1.004}. See \cite{huijse:10} for a review of these models. The observations here may thus provide  new tools for exploring aspects of superconformal field theories.

\subsection{Supersymmetric SYK and holography} A number of supersymmetric models of fermions on a graph have been constructed, starting with \cite{Nicolai:1976xp}. A generalization, which is relevant to the study of holography and black holes, is the supersymmetric version of the Sachdev-Ye-Kitaev (SYK) model \cite{Fu:2016vas}, corresponding to $\cP=\Bbb I$ and 
\begin{equation}
\cQ=i\sum_{i_{1}<i_{2}<\ldots<i_{q}} C_{i_{1}i_{2}\ldots i_{q}} \, a^{\dagger}_{i_{1}}a^{\dagger}_{i_{2}}\ldots a^{\dagger}_{i_{q}}\,,
\end{equation}
where each index takes values $i_{n}\in\{1,\ldots,N\}$  and $C_{i_{1}i_{2}\ldots i_{q}}$ is a totally antisymmetric tensor.  A refined Witten index can be computed exactly in this model, giving
\begin{equation}
\text{Tr}\, \left[ (-1)^{F} g^{r}\right] = e^{i N\pi (\frac{r}{q}-\frac12)}\left[2 \sin \frac{\pi r}{q}\right]^{N}\,,
\end{equation}
where $g=e^{i\pi n/q}$,  $n\in \Bbb N$, is a generator of  a discrete $Z_{q}$  symmetry of the system which commutes with  $\cQ$ and $\cQ^{\dagger}$,  $r\in \Bbb R$ is a chemical potential for this symmetry, and the trace is over all $N$-bit strings. This quantity is computed by setting $\hat U_{S}=g^{r}$ in the algorithm above. An interesting question is whether this quantum computational perspective on SYK and its refined Witten index can lead to further insights into the physics of supersymmetric black holes or AdS$_{2}$ holography (see \cite{Sarosi_2018,Rosenhaus:2018dtp} for reviews).  Although the exact generalized Witten index is known, and thus there is no need for an algorithm, this result may serve to benchmark the performance of quantum computers, perhaps along the lines of similar tests based on the Jones polynomial on IBM Q devices  \cite{gkta2019benchmarking}.

\end{document}